\documentclass[a4paper,UKenglish,cleveref,autoref]{lipics-v2021}

\title{Arc Kayles is PSPACE-complete}
\titlerunning{Arc Kayles is PSPACE-complete}

\author{\'{E}douard Bonnet}{CNRS, ENS de Lyon, Université Claude Bernard Lyon 1, LIP UMR 5668, Lyon, France \and \url{http://perso.ens-lyon.fr/edouard.bonnet}}{edouard.bonnet@ens-lyon.fr}{https://orcid.org/0000-0002-1653-5822}{}

\authorrunning{\'E. Bonnet}

\Copyright{Édouard Bonnet}

\category{}

\relatedversion{}

\supplement{}

\nolinenumbers 

\hideLIPIcs  

\EventEditors{John Q. Open and Joan R. Access}
\EventNoEds{2}
\EventLongTitle{42nd Conference on Very Important Topics (CVIT 2016)}
\EventShortTitle{CVIT 2016}
\EventAcronym{CVIT}
\EventYear{2016}
\EventDate{December 24--27, 2016}
\EventLocation{Little Whinging, United Kingdom}
\EventLogo{}
\SeriesVolume{42}
\ArticleNo{23}

\usepackage[utf8]{inputenc}  

\usepackage[T1]{fontenc}
\usepackage{lmodern}

\usepackage{amsmath}  
\usepackage{amssymb}
\usepackage{amsthm}
\usepackage{bbm}
\usepackage{accents}
\usepackage{complexity}

\usepackage{booktabs}
\usepackage{paralist}
\usepackage{bm}
\usepackage{fixmath}

\makeatletter
\newtheorem*{rep@theorem}{\rep@title}
\newcommand{\newreptheorem}[2]{%
\newenvironment{rep#1}[1]{%
 \def\rep@title{#2 \ref{##1}}%
 \begin{rep@theorem}}%
 {\end{rep@theorem}}}
\makeatother

\newreptheorem{theorem}{Theorem}
\newreptheorem{lemma}{Lemma}
\newreptheorem{corollary}{Corollary}

\usepackage{pgfplots}
\usepackage{xspace}
\usepackage{tikz}
\usepackage{tikz-3dplot}

\usepackage[ruled,vlined,linesnumbered]{algorithm2e}

\usetikzlibrary{fit} 
\usetikzlibrary{arrows}
\usetikzlibrary{arrows.meta}
\usetikzlibrary{patterns}
\usetikzlibrary{calc}
\usetikzlibrary{shapes}
\usetikzlibrary{positioning}
\usetikzlibrary{math}
\usetikzlibrary{shapes.symbols, shapes.geometric}
\usetikzlibrary{decorations.pathreplacing,calligraphy}
\usetikzlibrary{decorations.pathmorphing, backgrounds}
\usepackage[scr=boondox,scrscaled=1.05]{mathalfa}

\newcommand{\true}{\textsf{True}\xspace}
\newcommand{\false}{\textsf{False}\xspace}

\newcommand{\sg}{\mathsf{sg}}

\crefname{conjecture}{Conjecture}{Conjectures}
\crefname{claim}{Claim}{Claims}

\newcommand{\Oh}{\mathcal{O}}

\renewcommand{\P}{\mathcal{P}}
\newcommand{\N}{\mathcal{N}}

\DeclareMathOperator{\mex}{mex}

\begin{document}

\maketitle

\begin{abstract}
  We show that \textsc{Arc Kayles} is \PSPACE-complete.
  This solves a question raised by Schaefer in 1976.
\end{abstract}

\section{Introduction}\label{sec:intro}

In a~seminal work~\cite{Schaefer76,Schaefer78}, Schaefer introduced and studied several 2-player games on formulas and on graphs, such as \textsc{Pos CNF}, \textsc{Pos DNF}, \textsc{Node Kayles}, \textsc{Directed Generalized Geography}, and \textsc{Arc Kayles}, and showed the \PSPACE-completeness of all of them but \textsc{Arc Kayles}.
Despite active work and repeated attention to its complexity, this question had remained open for half a~century.

In this introduction, we only define \textsc{Node Kayles} and \textsc{Arc Kayles}.
In~\textsc{Node Kayles}, two players alternate removing the closed neighborhood of a~chosen remaining vertex.
Under the convention of \emph{normal play}, the first player without a~legal move loses (and the opponent wins).
In \textsc{Arc Kayles}, the two players alternate picking an edge and removing its endpoints.
Despite the name, it is a~game on \emph{undirected} graphs.

So, \textsc{Arc Kayles} is to maximal matchings what \textsc{Node Kayles} is to maximal independent sets.
The winner of \textsc{Arc Kayles} is the one playing the last edge of a~maximal matching (in the original graph), and that of \textsc{Node Kayles} is the one playing the last vertex of a~maximal independent set.

Polynomial-time algorithms were obtained for \textsc{Arc Kayles} on simple graph families~\cite{Huggan16}.
The parameterized complexity of this problem with respect to various structural parameters has been investigated~\cite{Hanaka24,Hanaka26}.
Recently, the misère partizan~\cite{Burke25} and the partizan \cite{Burke26b} versions of \textsc{Arc Kayles} were shown to be \PSPACE-complete.
Other variants of \textsc{Arc Kayles} have been considered~\cite{Dailly19,Burke26}.
Schaefer's paper~\cite{Schaefer78} also contained the \PSPACE-completeness of a~variant where the legal moves are restricted to a~prescribed set of edges (until this subset gets depleted).

We resolve the complexity of~\textsc{Arc Kayles}.

\begin{theorem}\label{thm:main}
  \textsc{Arc Kayles} is \PSPACE-complete.
\end{theorem}

In the \PSPACE-hardness proof of \textsc{Node Kayles} by Schaefer~\cite{Schaefer78}, a~\emph{cheating} move (one that does not correspond to something legal in the problem \textsc{Node Kayles} is reduced from) is immediately punished: an appropriate answer to it immediately wins the game. 
This is viable because the size of a~maximal independent set can be arbitrarily smaller than the size of a~maximum independent set.
This is not available for \textsc{Arc Kayles} since the minimum maximal matching and maximum matching sizes are within a~factor~2 of each other.

In other words, there cannot be short exits.
In principle, we will have to consider long sequences of possibly cheating moves (in effect, a~whole game tree of them).
This is intuitively why showing the hardness of \textsc{Arc Kayles} is a~more challenging task.
Indeed, the variants of \textsc{Arc Kayles} shown to be hard restrict the set of legal moves. 

\subparagraph*{Proof outline.}
A~crucial lemma is that the outcome of the game on any graph obtained from the \emph{biclique} $K_{a,b}$ (made by two fully-adjacent independent sets of size~$a$ and~$b$) by adding an independent set~$I$, such that every vertex of the biclique has a~pendant neighbor in~$I$, is fully determined by~$a$ and~$b$ only.
Actually, a~stronger property holds: the \emph{Sprague--Grundy value} (see~\cref{subsec:value}) only depends on~$a$ and~$b$. 

The reduction from \textsc{Pos CNF} (a~2-player game based on \textsc{SAT} where one player wants to satisfy a~negation-free instance and the opponent has the opposite goal) or \textsc{Maker--Breaker} (see~\cref{subsec:pos-cnf}) is built around that lemma.
If any edge incident to a~special vertex $s$ is played (and some mild technical assumptions hold), the graph breaks into connected components whose values can be determined by the lemma.
The part of the graph that actually depends on the \textsc{Pos CNF} instance (and not merely on its number of variables and clauses) is between a~side of the biclique (the clauses) and a~part of the independent set (the variables); see~\cref{fig:main}.
A~parity mechanism incentivizes the second player (but not the first player) to deprive a~clause vertex of its literal vertices, after which the second player can win by playing the edge incident to $s$ corresponding to the unsatisfied clause.

The winning strategy for the first player initially simulates the winning strategy to satisfy all the clauses.
After this phase is completed, follows a~more complicated ``endgame'' where several parities and Sprague--Grundy values fortunately align.

\subparagraph*{Open questions.}
While now settled on general graphs, the complexity of \textsc{Arc Kayles} remains mysterious on restricted inputs.
We do not know whether a~polynomial-time algorithm exists for trees, for bounded-degree graphs, for planar graphs, etc.
These questions happen to also be open for \textsc{Node Kayles}.

\section{Preliminaries}\label{sec:prelim}

For any two integers $i, j$, we set $[i,j] := \{k \in \mathbb Z \mid i \leqslant k \leqslant j\}$, and $[i] := [1,i]$.
We denote by $\mathbb N_0$ the set of nonnegative integers.
We use the standard graph-theoretic notation.
For any graph $G$, $V(G)$ and $E(G)$ denote the vertex set and edge set, respectively, of~$G$.
If $S \subseteq V(G)$, then $G[S]$ denotes the subgraph of~$G$ induced by~$S$, and $G-S := G[V(G) \setminus S]$.
An \emph{independent set} is a~set of vertices that are pairwise nonadjacent. 

\subsection{Sprague--Grundy values}\label{subsec:value}

We denote by $\sg(G)$ the \emph{Sprague--Grundy value} (or \emph{value}, for short) of the \textsc{Arc Kayles} instance~$G$.
It is inductively defined as
\[\sg(G) := \mex\{ \sg(G-\{u,v\}) \mid uv \in E(G)\},\]
where $\mex$ is the minimum excluded nonnegative integer, that is, \[ \mex(S) := \min \{n \in \mathbb N_0 \mid n \notin S\}.\]
In particular, when $G$ is edgeless, $\sg(G)=0$.
Let us first highlight the following property of the value function.

\begin{observation}\label{obs:value}
  From any position of value $n \in \mathbb N_0$, and for any nonnegative integer $i < n$, there is a~move to a~position of value $i$, but there are no moves to a~position of value~$n$.
\end{observation}

A~graph $G$ is a~\emph{$\P$-position} (implicitly for \textsc{Arc Kayles}) if the \textbf{p}revious player (the one not to move) has a~winning strategy, and it is an~\emph{$\N$-position} if the \textbf{n}ext player (the one to move) has a winning strategy.
In particular, an edgeless graph is a~$\P$-position, and a~single edge is an $\N$-position. 
As \textsc{Arc Kayles} on finite graphs has finite game trees, it is determined.
So every graph is either a~$\P$-position or an~$\N$-position.

We will assume that isolated vertices are automatically removed from the current position.
This does not affect the game, and simplifies the exposition.
For instance, there is now a~unique edgeless or terminal position: the empty graph.

The $\P$-positions are characterized by the Sprague--Grundy value.

\begin{lemma}\label{lem:charac-sg}
  A~graph $G$ is a~$\P$-position if and only if $\sg(G)=0$.
\end{lemma}
\begin{proof}
  The winning strategy is to play (when possible) an edge such that the resulting position has value 0.
  By definition of $\sg$, this is possible if and only if the current position has positive value.
  Indeed, if $\sg(G) > 0$, then $0 \in \{ \sg(G-\{u,v\}) \mid uv \in E(G)\}$; simply play any $uv \in E(G)$ such that $\sg(G-\{u,v\})=0$.
  On the other hand, if $\sg(G)=0$, every move $uv \in E(G)$ is such that $\sg(G-\{u,v\}) > 0$.

  Following this strategy, the positions alternate between having value 0 and having positive value.
  This is thus winning for the player who started if the initial position has positive value, since the unique terminal position has value 0.
  If the initial position has value 0, then after any move, the opponent can apply the same winning strategy.
\end{proof}

Given two graphs $G_1, G_2$, we denote by $G_1 \uplus G_2$ their disjoint union.
Given two numbers $a, b \in \mathbb N_0$, we denote by $a \oplus b$ the integer whose binary expansion is the bitwise xor of the binary expansion of~$a$ and the binary expansion of~$b$.
For instance, $5 \oplus 6 = 3$, because the bitwise xor of $101$ and $110$ is $011$.
Both $\uplus$ and $\oplus$ are associative.

The value of a~graph is determined by the values of its connected components.
This is a~consequence of the Sprague--Grundy theorem, first shown by Bouton~\cite{Bouton1901} for the special case of the \textsc{Nim} game, and independently generalized by~Sprague~\cite{Sprague35} and Grundy~\cite{Grundy39} to any impartial game.

\begin{lemma}[\cite{Sprague35,Grundy39}]\label{lem:union}
  For all graphs $G_1, \ldots, G_h$, $\sg(G_1 \uplus \cdots \uplus G_h) = \sg(G_1) \oplus \cdots \oplus \sg(G_h)$.
\end{lemma}
\begin{proof}
  It is enough to show the statement when $h=2$, by associativity of $\uplus$ and $\oplus$.
  We show the lemma by induction on $|V(G_1)|+|V(G_2)|$.
  If both $G_1$ and $G_2$ are empty, then $G_1 \uplus G_2$ is empty, and $\sg(G_1) \oplus \sg(G_2) = 0 \oplus 0 = 0$, which is indeed equal to $\sg(G_1 \uplus G_2)$.

  We now assume that the lemma's equality holds in every strict induced subgraph of $G_1 \uplus G_2$.
  In particular, for every $uv \in E(G_1)$, \[\sg((G_1 - \{u,v\}) \uplus G_2) = \sg(G_1 - \{u,v\}) \oplus \sg(G_2),\] and for every $uv \in E(G_2)$, \[\sg(G_1 \uplus (G_2 - \{u,v\})) = \sg(G_1) \oplus \sg(G_2 - \{u,v\}).\]
  Thus the set of reachable values after playing one move in $G_1 \uplus G_2$ is \[ S := \{\sg(G_1 - \{u,v\}) \oplus \sg(G_2) \mid uv \in E(G_1)\} \cup \{\sg(G_1) \oplus \sg(G_2 - \{u,v\}) \mid uv \in E(G_2)\}.\]

  We first claim that $\sg(G_1) \oplus \sg(G_2) \notin S$.
  Indeed, this would imply that either $\sg(G_1 - \{u,v\}) \oplus \sg(G_2) = \sg(G_1) \oplus \sg(G_2)$ for some $uv \in E(G_1)$, hence $\sg(G_1 - \{u,v\}) = \sg(G_1)$, or symmetrically that $\sg(G_2 - \{u,v\}) = \sg(G_2)$ for some $uv \in E(G_2)$.
  This cannot happen by~\cref{obs:value}.

  We finally claim that every nonnegative integer $i < \sg(G_1) \oplus \sg(G_2)$ satisfies $i \in S$.
  We assume that $\sg(G_1) \oplus \sg(G_2) > 0$, since otherwise we are done.
  Let $k$ be the index of the heaviest (nonzero) bit of $\sg(G_1) \oplus \sg(G_2)$ (counting from 1 for the lightest bit), and $k' \leqslant k$ be the index of the heaviest bit on which $\sg(G_1) \oplus \sg(G_2)$ has a~1 whereas $i$ has a~0.
  By assumption, exactly one of $\sg(G_1), \sg(G_2)$, say, $\sg(G_1)$, has a~1 at index~$k'$.
  By~\cref{obs:value}, one can thus play in~$G_1$ a~move that leads to a~position of value $j$, such that the binary expansion of~$j$ agrees with that of~$\sg(G_1)$ on the indices larger than $k'$, has a~0 at index~$k'$, and the pointwise xor of $\sg(G_2)$ and $i$ at indices smaller than~$k'$.
  Then $j < \sg(G_1)$ and $j \oplus \sg(G_2)=i$.
  Hence $i \in S$, as required.
\end{proof}

\subsection{The \textsc{Pos CNF} game, outcome-equivalent games}\label{subsec:pos-cnf}

 To show the main result, we will reduce from the \textsc{Pos CNF} game.
  Given an $n$-variable CNF formula $\varphi = C_1 \land \ldots \land C_m$ where every clause $C_j$ has only positive literals, two players \true and \false alternate, with \true playing first, claiming an unclaimed variable.
  \true sets her variables to true, and \false sets his variables to false.
  \true wins if all clauses of $\varphi$ are eventually satisfied, and loses otherwise. 
  Deciding if \true has a~winning strategy, the \textsc{Pos CNF} problem, is indeed \PSPACE-hard~\cite{Schaefer78}.
  Equivalently, this is the Breaker-first \textsc{Maker--Breaker} game on the hypergraph whose vertices are the variables and whose hyperedges are the clauses, with \true playing Breaker and \false playing Maker.

  We say that two combinatorial games are \emph{outcome-equivalent} if the first player (resp. second player) has a~winning strategy from the same set of positions.
  We show the intuitive fact that passing or playing useful moves for the opponent never helps in \textsc{Maker--Breaker}.
  
\begin{lemma}\label{lem:pb-monotonicity}
   The \textsc{Pos CNF} game is outcome-equivalent to its variant where players can additionally pass or assign either truth value to the selected variable.
\end{lemma}
\begin{proof}
  We use the \textsc{Maker--Breaker} phrasing, where Breaker plays the first move.
  Maker (\false) has to play all the elements of some hyperedge, in order to win.
  Now, players can pass or can claim an element for the benefit of their opponent.
  If the player $P$ to move has a~winning strategy in \textsc{Maker--Breaker}, then $P$ simply follows this strategy.   

  Now consider any position such that the player $P$ not to move has a~winning strategy in \textsc{Maker--Breaker}.
  If the opponent plays a~legal move of \textsc{Maker--Breaker}, then $P$ again follows the winning strategy.
  If the opponent passes or claims an element for $P$, then $P$ claims any unclaimed element.
  In the resulting position, $P$ has one or two extra claimed vertices, while the opponent has no extra claimed vertices.
  So, $P$ still has a~winning strategy.
  Indeed, when the winning strategy without the extra claimed vertices says that $P$ should play one of these vertices, $P$ plays any unclaimed vertex instead.  
\end{proof}

  
\section{Hardness of Arc Kayles}

We start with a~lemma determining the Sprague--Grundy value of any biclique neighboring an independent set, with every vertex of the biclique having at least one private neighbor (that is, in that case, a~pendant neighbor).

\begin{lemma}\label{lem:biclique}
  Let $H$ be any graph where $V(H)$ is partitioned into two sets $B, I$ such that:
  \begin{compactitem}
    \item $H[B]$ is an induced biclique isomorphic to $K_{a,b}$ for some nonnegative integers $a, b$,
    \item $I$ is an independent set, and
    \item every vertex of~$B$ has a~neighbor in~$I$ that has degree 1 (in~$H$).
  \end{compactitem}
  Then \[\sg(H) = g(a,b) := {(a+b) \bmod 2}\,+\,2 \cdot ({\min(a,b) \bmod 2}).\]
\end{lemma}
\begin{proof}
  The assumptions are such that the \textsc{Arc Kayles} game on~$H$ can be thought of as playing on $K_{a,b}$ the more permissive game where in addition to playing an edge, one can play a~vertex (in which case, only the played vertex is removed).
  In the latter game on $K_{a,b}$, isolated vertices (which arise if a~side of the biclique becomes empty) are \emph{not} removed.

  Indeed, while a~vertex $v \in B$ remains, one can simulate playing $v$ by playing an edge linking $v$ to a~neighbor of~$v$ in $I$ with degree 1.
  This neighbor cannot disappear before~$v$.
  Furthermore, every edge of the biclique is playable as a~normal edge move, and every other edge has exactly one endpoint in~$B$ and corresponds to a~vertex move.

  We show the lemma by induction on $a+b$.
  If $a=0$ or $b=0$, every move deletes exactly one vertex of~$B$, and the other vertices of~$B$ remain non-isolated due to their private neighbor in~$I$ (one with degree~1).
  Thus $\sg(H)= b \mod 2 = g(0,b)$ when $a=0$, and $\sg(H) = a \mod 2 = g(a,0)$ when $b=0$.
  Indeed, when the remaining number of vertices in~$B$ is odd, every legal move wins (thus leads to a~position of value 0), and when this number is even the position is losing.

  By the first paragraph and the induction assumption, when $a, b > 0$,
  \[\sg(H) = \mex \{g(a-1,b), g(a,b-1), g(a-1,b-1)\}.\]
  We show that this minimum excluded value is indeed equal to~$g(a,b)$.

  If $a \neq b$, we claim that \[\{g(a,b), g(a-1,b), g(a,b-1), g(a-1,b-1)\}=\{0,1,2,3\},\]
  which is sufficient to conclude.
  Say, without loss of generality, that $a<b$.
  Then, compared to the binary expansion of $g(a,b)$, the numbers $g(a-1,b), g(a,b-1), g(a-1,b-1)$ respectively flip both bits, flip the first bit, flip the second bit.

  If $a=b$, then $\{g(a-1,b), g(a,b-1), g(a-1,b-1)\}=\{0,1\}$ when $a$ is odd, and $\{g(a-1,b), g(a,b-1), g(a-1,b-1)\}=\{2,3\}$ when $a$ is even.
  Furthermore, $g(a,b)=g(a,a)=2(a \mod 2)$, so we indeed have $g(a,b)=2$ in the former case, and $g(a,b)=0$ in the latter.
\end{proof}

Interestingly, the game in \cref{lem:biclique} with a~biclique $K_{a,b}$ is equivalent to Berlekamp's \emph{impartial King game} on an~$(a+1) \times (b+1)$ chess board.
The Sprague--Grundy values of the latter game as a~function of~$a$ and~$b$ are established in~\cite[Theorem 3.8]{Gottlieb25}, while crediting a~proof illustration by Berlekamp~\cite{Berlekamp17}.

We now show the main result.

\begin{reptheorem}{thm:main}
  \textsc{Arc Kayles} is \PSPACE-complete.
\end{reptheorem}

\begin{proof}
  \textsc{Arc Kayles} is in \PSPACE.
  Indeed, a~game on~$G$ lasts at~most $N/2$ moves with $N := |V(G)|$, and a~single position can be described with $\Oh(N^2)$ bits.
  Thus a~depth-first search of the game tree, only keeping the positions from the root to the current position, solves the game in space $\Oh(N^3)$. 

  We reduce from the \PSPACE-hard \textsc{Pos CNF} problem~\cite{Schaefer78}.
  Let $X=\{x_1, \ldots, x_n\}$ be the variable set of any instance $\varphi = C_1 \land \cdots \land C_m$.
  We further assume that $m$ is odd.
  We simply duplicate one clause if this is not the case, as this does not change which player has a~winning strategy.

  \medskip
  \textbf{Construction.}
  We set \[K := 4n+6~~\text{and}~~R := m+2n+2.\]
  
  We build a~graph $G_\varphi := G$ with vertex set:
  \[\{s,t\} \cup \{a_j, b_j \mid j \in [m]\} \cup \{v_i, t_i \mid i \in [R]\} \cup \{f_i \mid i \in [n]\} \cup \{y_i, z_i \mid i \in [K]\},\]
  and edge set:
  \[\{st\} \cup \{sa_j, a_jb_j \mid j \in [m]\} \cup \{b_jv_i \mid j \in [m], i \in [R]\} \cup \{v_it_i \mid i \in [R]\} \cup \{v_if_i \mid i \in [n]\}\]
  \[\cup~\{b_jf_i \mid x_i \in C_j\} \cup \{sy_i, y_iz_i \mid i \in [K]\}.\]
  This finishes the construction; see~\cref{fig:main}.
  Graph $G$ has $2+2m+2K+2R+n = \Oh(n+m)$ vertices, and can be built in polynomial time.

\begin{figure}[!ht]
\centering
\begin{tikzpicture}[
  vertex/.style={circle,fill,inner sep=1.3pt},
  group/.style={draw,rectangle,rounded corners,
                inner xsep=4pt,inner ysep=3pt,
                minimum height=7.5mm},
  dimension/.style={<->,>=Stealth,thin},
  pass/.style={draw=blue,line width=1.2pt},
  assignment/.style={draw=red,line width=1.2pt},
  every node/.style={font=\small},
  line width=.4pt
]

\node[vertex,label=left:$s$] (s) at (.3,0) {};
\node[vertex,label=right:$t$] (t) at (.3,.95) {};

\foreach \i/\x in {1/-3.4,2/-2.8,j/-1.6}{
  \node[vertex] (a\i) at (\x,-1.45) {};
  \node[vertex] (b\i) at (\x,-2.75) {};
}
\foreach \i/\x in {1/2.8,2/3.4,K/4.6}{
  \node[vertex] (y\i) at (\x,-1.45) {};
  \node[vertex] (z\i) at (\x,-2.75) {};
}
\foreach \y in {-1.45,-2.75}{
  \node at (-2.2,\y) {$\cdots$};
  \node at (4,\y) {$\cdots$};
}

\foreach \i/\x in {1/-5.2,2/-4.6,i/-4.0,n/-2.6,npone/-1.2,R/.4}{
  \node[vertex] (v\i) at (\x,-4.5) {};
  \node[vertex] (t\i) at (\x,-6.25) {};
}
\node at (-3.3,-4.5) {$\cdots$};
\node at (-3.3,-6.25) {$\cdots$};
\node at (-.4,-4.5) {$\cdots$};
\node at (-.4,-6.25) {$\cdots$};

\foreach \i/\x in {1/2.5,2/3.1,i/3.7,n/4.9}
  \node[vertex] (f\i) at (\x,-6.25) {};
\node at (4.3,-6.25) {$\cdots$};

\foreach \name/\first/\last in {A/a1/aj,B/b1/bj,V/v1/vR,T/t1/tR}
  \node[group,fit=(\first)(\last),label=left:$\name$] (\name) {};
\foreach \name/\first/\last in {Y/y1/yK,Z/z1/zK,F/f1/fn}
  \node[group,fit=(\first)(\last),label=right:$\name$] (\name) {};

\node[right=2pt,inner sep=0pt] at (bj) {$b_j$};

\node[above=2pt,inner sep=0pt] at (vi) {$v_i$};
\node[above=2pt,inner sep=0pt] at (vn) {$v_n$};
\node[above=2pt,inner sep=0pt] at (vnpone) {$v_{n+1}$};
\node[above=2pt,inner sep=0pt] at (vR) {$v_R$};

\node[below=2pt,inner sep=0pt] at (ti) {$t_i$};
\node[below=2pt,inner sep=0pt] at (tn) {$t_n$};
\node[below=2pt,inner sep=0pt] at (tnpone) {$t_{n+1}$};
\node[below=2pt,inner sep=0pt] at (tR) {$t_R$};

\node[below=2pt,inner sep=0pt] at (fi) {$f_i$};
\node[below=2pt,inner sep=0pt] at (fn) {$f_n$};

\begin{scope}[on background layer]
  \draw (s)--(t);

  \foreach \i in {1,2,j}
    \draw (s)--(a\i)--(b\i);

  \foreach \i in {1,2,K}{
    \draw (s)--(y\i);
    \draw[pass] (y\i)--(z\i);
  }

  \draw[line width=1pt] (B.south)--(V.north);

  \foreach \i in {1,2,i,n}{
    \draw[assignment] (v\i)--(t\i);
    \draw[assignment] (v\i)--(f\i);
  }

  \draw[pass] (vnpone)--(tnpone);

  \draw[pass] (vR)--(tR);

  \draw (bj) .. controls +(0.9,-1.35) and +(0,1.65) ..
    node[pos=.5,above,sloped,inner sep=3pt] {$x_i\in C_j$} (fi);
\end{scope}

\draw[dimension] ([yshift=11mm]A.west) --
  node[above] {$m$} ([yshift=11mm]A.east);

\draw[dimension] ([yshift=-11mm]t1.west) --
  node[below] {$n$} ([yshift=-11mm]tn.east);

\draw[dimension] ([yshift=-7.5mm]Z.west) --
  node[below] {$K$} ([yshift=-7.5mm]Z.east);
\draw[dimension] ([yshift=-16mm]T.west) --
  node[below] {$R$} ([yshift=-16mm]T.east);
\draw[dimension] ([yshift=-7.5mm]F.west) --
  node[below] {$n$} ([yshift=-7.5mm]F.east);

\end{tikzpicture}
\caption{The graph $G_\varphi$.
  Pass moves are blue and assignment moves are red.
  The thick black edge represents the biclique between $B$ and $V$.}
\label{fig:main}
\end{figure}
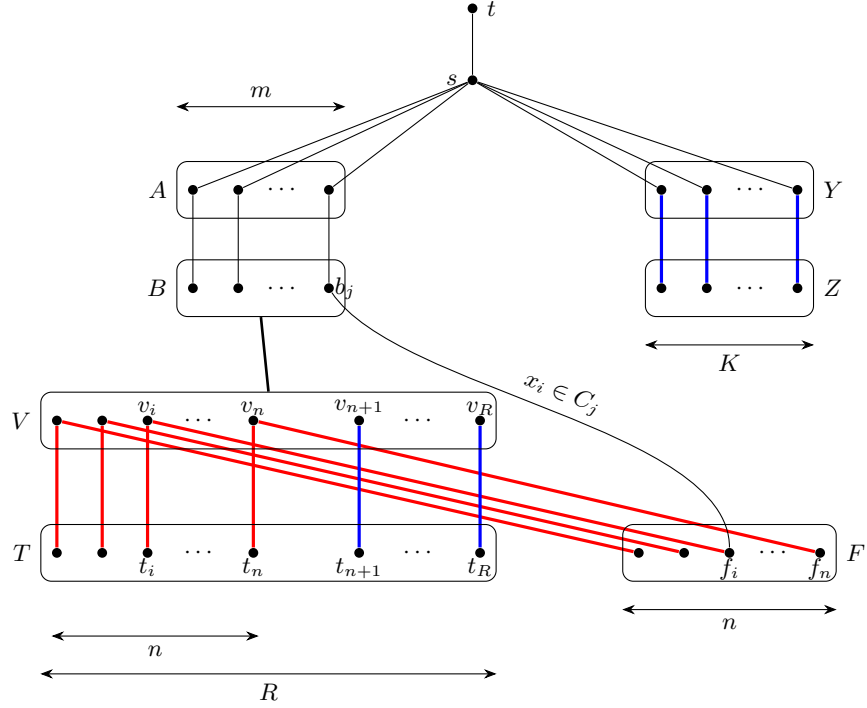

 We further set $A := \{a_j \mid j \in [m]\}$, $B := \{b_j \mid j \in [m]\}$, $V := \{v_i \mid i \in [R]\}$, $T := \{t_i \mid i \in [R]\}$, $F := \{f_i \mid i \in [n]\}$, $Y := \{y_i \mid i \in [K]\}$, and $Z := \{z_i \mid i \in [K]\}$.
  All these sets are independent sets.
  While vertices are removed by the game (by being endpoints of a~played edge or becoming isolated), we keep denoting by $A, B, V, T, F, Y, Z$ their respective intersections with the current vertex set.
  
  For every $i \in [n]$, playing the edge $v_it_i$ is interpreted as setting $x_i$ to true, while playing $v_if_i$ is interpreted as setting $x_i$ to false.
  We call any move $v_it_i$ or $v_if_i$, with $i \in [n]$, an \emph{assignment move}.
  We call any move $v_it_i$, with $i \in [n+1,R]$, or $y_iz_i$, with $i \in [K]$, a~\emph{pass move}. 

  In any given position, we automatically remove the isolated vertices (this does not affect the game), and denote by $r \leqslant R$ the number of remaining vertices in $V$, and by $k \leqslant K$ the number of remaining vertices in $Y$.    
  
  \begin{claim}\label{clm:deviation}
    While $r \geqslant m$ and $k \geqslant 1$, the first move that is neither an assignment move nor a~pass move loses, except if it is $sa_j$ and all the literals of~$C_j$ in $F$ are gone.  
  \end{claim}
  \begin{claimproof}
    We use \cref{lem:union,lem:biclique} to show that deviating from assignment or pass moves loses.

    \subparagraph*{An edge $\bm{st}$ or $\bm{sy_i}$ is played.}
    Such a move (any move that removes $s$) disconnects $Y \cup Z$ from $A \cup B \cup V \cup T \cup F$.
    We can thus apply \cref{lem:biclique} to the latter connected component, with the biclique $(B,V)$ and the independent set $A \cup T \cup F$.
    Note indeed that all the previous assignment and pass moves maintain that every vertex $b_j \in B$ has a~private neighbor $a_j \in A$, and every (remaining) vertex $v_i \in V$ has a~private neighbor $t_i \in T$.

    Therefore the value of the resulting graph is $g(m,r) \oplus (k \mod 2)$ if $st$ was played, and $g(m,r) \oplus ((k-1) \mod 2)$ if some $sy_i$ was played.
    Indeed, the subgraph induced by the remaining vertices of~$Y \cup Z$ consists of $k$ (resp.~$k-1$) isolated edges, and the value $k \mod 2$ (resp.~$(k-1) \mod 2$) follows from~\cref{lem:union}.

    Since $m \leqslant r$ and $m$ is odd, $\min(m,r) \mod 2=1$, so $g(m,r) \in \{2,3\}$.
    Thus \[g(m,r) \oplus 0,~g(m,r) \oplus 1 \in \{2,3\},\] so the resulting graph is an~$\N$-position.

    \subparagraph*{An edge incident to some $\bm{b_j \in B}$ is played.}
    Such a~move can be $a_jb_j$, $b_jv_i$ for some remaining $v_i \in V$, or $b_jf_i$ for some remaining $f_i \in F$ (such that $x_i \in C_j$).

    In all cases, the smallest side of the remaining biclique $(B,V)$ has size $m-1$, which is even.
    And by the analysis of the previous case, if $st$ is played next, the resulting value is \[g(m-1,r') \oplus (k \mod 2),\]
    whereas if $sy_i$ is played next, the resulting value is \[g(m-1,r') \oplus ((k-1) \mod 2),\]
    where $r' \in \{r-1,r\}$.
    In any case, $g(m-1,r') \in \{0,1\}$ since $\min(m-1,r') \mod 2 = (m-1) \mod 2 = 0$.
    So exactly one of $st, sy_i$ is a winning move: it leads to a~position with value~0.
    There is at least one remaining edge $sy_i$ by assumption that $k \geqslant 1$.

    \subparagraph*{An edge $\bm{sa_j}$ is played while some edge $\bm{b_jf_i}$ remains.}
    Against this move, both replies $b_jv_h$ (for some remaining $v_h \in V$) and $b_jf_i$ allow one to then invoke~\cref{lem:biclique}.
    These replies lead to positions of value $g(m-1,r-1) \oplus (k \mod 2)$ and $g(m-1,r) \oplus (k \mod 2)$, respectively.
    Since $\min(m-1,r-1)=m-1$ and $m-1$ is even,
    \[\{g(m-1,r-1), g(m-1,r)\}=\{(m+r) \mod 2,~(m+r-1) \mod 2\} = \{0,1\}.\]
    Therefore exactly one of $b_jv_h, b_jf_i$ is a~winning reply (it leads to a~$\P$-position).

  \medskip  
  This concludes the proof of the claim since every other edge is an assignment move, a~pass move, or is covered by the exception.  
  \end{claimproof}

  \begin{claim}\label{clm:empty-clause}
    If only assignment and pass moves have been played, $r \geqslant m$, and a~vertex $b_j \in B$ has no neighbor left in~$F$, then the move $sa_j$ is winning if and only if $r+k$ is even.
  \end{claim}
  \begin{claimproof}
    After the move $sa_j$ we can apply \cref{lem:biclique} to the remaining biclique $(B \setminus \{b_j\}, V)$ and independent set $A \cup T \cup F \cup \{b_j\}$, leveraging the fact that $b_j$ has no remaining neighbor outside of~$V$, so can be transferred from the biclique to the independent set.
    This results in a~position of value $g(m-1,r) \oplus (k \mod 2)=(r+k) \mod 2$, since $m-1 \leqslant r$ and $m-1$ is even.
  \end{claimproof}

  We are now equipped to show that the reduction is correct.
  This is done in the following two claims.

  \begin{claim}\label{clm:false-wins}
    If \false has a~winning strategy on~$\varphi$, then the \textsc{Arc Kayles} instance $G_\varphi$ is winning for the second player~$O$.
  \end{claim}

  \begin{claimproof}
    The second player mimics \false's strategy, and only plays assignment moves of the form $v_if_i$ until
    \begin{compactitem}
      \item a~vertex $b_j \in B$ has no neighbor in~$F$, or
      \item the first player $P$ plays a~losing move according to~\cref{clm:deviation} (and \cref{clm:empty-clause}).
    \end{compactitem}

    By~\cref{lem:pb-monotonicity}, $P$ cannot prevent the first condition from eventually holding by only playing assignment or pass moves.
    By~\cref{clm:deviation}, the only other type of moves that $P$ can play is $sa_j$ exactly when the first item holds for $b_j$ (as a~consequence of $O$'s last move).
    This happens by move number $2n+1$, at the latest, since $O$ only plays assignment moves.

    At this point, the number of pass or assignment moves $R-r+K-k$ is even, since $P$ and $O$ have played an equal number of such moves.
    Therefore $r+k$ has the same parity as $R+K = m + 6n + 8$, which is odd.
    Moreover, $r \geqslant m+2 \geqslant m$ and $k \geqslant 3n+6 \geqslant 1$.
    So, \cref{clm:deviation} indeed still applies and \cref{clm:empty-clause} implies that $sa_j$ is losing ($r+k$ is odd).

    Therefore, the next move of~$P$ has to be another assignment or pass move, to not be losing by~\cref{clm:deviation}.
    Then $O$ wins by playing $sa_j$, since now $r+k$ is even.
  \end{claimproof}

  \begin{claim}\label{clm:true-wins}
    If \true has a~winning strategy on~$\varphi$, then the \textsc{Arc Kayles} instance $G_\varphi$ is winning for the first player~$P$.
  \end{claim}

  \begin{claimproof}
   Player $P$ mimics \true's strategy, and unless $O$ plays a~losing move by~\cref{clm:deviation}, only plays assignment moves of the form $v_it_i$ until the first item holds, and then plays moves of the form $v_it_i$ (assignment or pass) until the second item holds:
    \begin{compactitem}
      \item for every $b_j \in B$ there is an $i$ such that $x_i \in C_j$ and the edge $v_it_i$ has been played, and
      \item $r=m$.
    \end{compactitem}
    
    By~\cref{clm:deviation}, while $r \geqslant m$ and $k > 0$, player $O$ can only play assignment or pass moves.
    The exceptional moves of \cref{clm:deviation} are prevented by the winning strategy of \true, which implies that for every $b_j \in B$, there is an~$i$ such that $x_i \in C_j$ and $f_i$ remains in~$F$.
    Therefore, by~\cref{lem:pb-monotonicity}, $P$ achieves the first item after at~most $2n-1$ moves.

    Then, $P$ continues to play moves of the type $t_iv_i$ until the second item is also fulfilled.
    This happens by move $4n+4$, at the latest.
    Thus, $k > 0$ and \cref{clm:deviation} applies, forcing $O$ to play assignment or pass moves, until both items are satisfied.

    At this stage, a~further move $v_it_i$ or $v_if_i$ would leave $r=m-1$ and lose to exactly one of $st$ or some remaining $sy_h$.
    Indeed, recall that $g(m,m-1) \in \{0,1\}$ since $\min(m,m-1)=m-1$ is even.
    As a~consequence, by~\cref{clm:deviation}, every move other than an edge $y_hz_h$ loses, and we may assume that both players exhaust the edges between $Y$ and~$Z$.

    This phase ends when $k=0$ (all the edges of $Y \cup Z$ have been played).
    At this point, $K+R-m=6n+8$ moves have been played.
    This is an even number, so it is $P$'s turn.
    Player $P$ then plays any remaining edge $v_it_i$ (after that $r=m-1$), and for the last phase of the game maintains the following invariants, when it is $O$'s turn:
    \begin{itemize}
    \item $r$ is even,
    \item $q$, defined as the remaining number of vertices in $B$, is odd,
    \item $r<q$, and
    \item every remaining vertex of~$B$ has at least one neighbor left in~$F$.
    \end{itemize}

    Note that they initially hold, since $r=m-1 < q = m$.
    We now consider all the types of moves for $O$.
    For each, we show how $P$ replies to maintain the invariants or win the game with \cref{lem:biclique}.

    \subparagraph*{$\bm{O}$ plays an edge $\bm{v_it_i}$ or $\bm{v_if_i}$.}
    Then $P$ replies with another edge $v_ht_h$.
    Such an edge always remains since just after $O$'s move, $r$ is odd.
    This decreases $r$ by~2, and does not change $q$.
    Even when $O$ plays an edge $v_if_i$, the fourth item still holds, since for every $b_j \in B$, in the first stage of the game (when $P$ simulated \true's winning strategy for the \textsc{Pos CNF} instance $\varphi$) an edge $v_ht_h$ with $x_h \in C_j$ was played.
    This makes the move $v_hf_h$ unavailable to~$O$.
    So the invariants are indeed maintained.

    \subparagraph*{$\bm{O}$ plays an edge $\bm{b_jv_i}$.}
    Then $P$ replies with another edge $b_{j'}v_{i'}$, which still exists since $r$ is odd after $O$'s move.
    This decreases $r$ and $q$ by~2.
    Again, the four invariants are preserved.

    \medskip

    We now show that all other moves from $O$ lose by~\cref{lem:biclique} (even when $r=0$).
    
    \subparagraph*{$\bm{O}$ plays the edge $\bm{st}$ or an edge $\bm{sa_j}$ for an already removed $\bm{b_j}$.}
    After such a~move, \cref{lem:biclique} applies with the remaining biclique $(B,V)$ and independent set $A \cup T \cup F$.
    The value of the position after this move is
    \[g(q,r) = (q+r) \mod 2 + 2 \cdot (\min(q,r) \mod 2) = 1 + 2 \cdot (r \mod 2)=1,\]
    since, by the invariants, $r$ is even, $q$ is odd, and $\min(q,r)=r$.
    Thus, the position is an $\N$-position, and $P$, now to move, has a~winning strategy.

    \subparagraph*{$\bm{O}$ plays an edge $\bm{sa_j}$ for a~remaining $\bm{b_j}$.}
    Then $P$ replies with an edge $b_jf_i$ (for the same $j$), which exists, by the fourth invariant.
    After $P$'s reply, \cref{lem:biclique} applies with the remaining biclique $(B,V)$ and independent set $A \cup T \cup F$.
    We also have that $q$ is now even, like~$r$.
    So the value of the position is $g(q,r) = 0$.
    Therefore it is a~$\P$-position, and $O$ loses.

    \subparagraph*{$\bm{O}$ plays an edge $\bm{a_jb_j}$ or $\bm{b_jf_i}$.}
    Then $P$ replies with the edge $st$.
    Again, \cref{lem:biclique} applies with $q$ and $r$ both even, so $g(q,r)=0$, and $O$ loses.

    \medskip

    There are no other moves for $O$.
    The game therefore ends in the following way.
    While $O$ plays an edge incident to $V$, player $P$ maintains the invariants, until possibly $r=0$.
    The first time $O$ plays an edge that is not incident to~$V$ (possibly, because $r=0$, so there are no edges incident to~$V$ left), $P$ has a~winning reply. 
  \end{claimproof}
  We conclude by~\cref{clm:false-wins,clm:true-wins}.
\end{proof}

Perhaps counterintuitively, the same reduction with $m \geqslant 3$ also shows the \PSPACE-hardness of the misère variant of \textsc{Arc Kayles}, where the first player without a~legal move wins. 

\begin{theorem}\label{thm:misere}
  \textsc{Misère Arc Kayles} is \PSPACE-complete.
\end{theorem}

\begin{proof}[Proof sketch]
  One should first observe that the winning positions among the graphs of~\cref{lem:biclique} are the same under misère and normal play when $a>0$ and $b>0$, and are flipped when $a=0$ or $b=0$.
  One way to do so is to follow Berlekamp's filling of the chessboard~\cite{Berlekamp17} and only register whether the positions are winning (under the misère play).
  Furthermore, for every graph~$H$ from~\cref{lem:biclique} and integer $k \geqslant 0$, $H \uplus kK_2$ (the disjoint union of $H$ and $k$ isolated edges) has the same outcome under misère and normal play when $a,b > 0$, and, when $a=0$ or $b=0$, its misère outcome is losing if and only if $a+b+k$ is odd. 

  Therefore, the winning strategy for player $O$ essentially does not change: it leads to a~winning application of~\cref{lem:biclique} inside of which the winning moves are different (but not the winning positions).
  The winning strategy of player $P$ is the same up to the point where $r=2$.
  If $O$ plays a~$V$--$T$ or $V$--$F$ edge, $P$ answers with a~$B$--$V$ edge, and if $O$ plays instead a~$B$--$V$ edge, then $P$ replies with a~$V$--$T$ edge.
  After these two moves, $r=0$ and $q$ is \emph{even}.

  Now, we claim that $P$ wins.
  If $O$ plays $st$ or $sa_j$ such that $b_j$ is already gone, every subsequent move will remove exactly one vertex in~$B$ (and $P$ wins because $q$ is even).
  If $O$ instead plays some edge $a_jb_j$ or $b_jf_i$, then $P$ plays $st$ and wins for the same reason.
  Finally, if $O$ instead plays $sa_j$ for some remaining $b_j$, then $P$ plays an edge $b_jf_i$ (which still exists) and wins. 
\end{proof}

\subparagraph*{AI disclosure.}
The proof of~\cref{thm:main} was obtained as the result of a~session with GPT-5.6 Pro, continued with GPT-6 Pro.
The technical contributions of the author during this discussion were light.
His suggestion of articulating the reduction around a~half-graph overshot the eventual use of a~mere biclique (\cref{lem:biclique}).
Indeed, a~half-graph alone is already too complicated: we do not know its Sprague--Grundy value as a~function of its height. 
The author's main contribution was to write a~self-contained proof that does not require any specific knowledge from the reader.
The write-up is entirely due to the author. 

\bibliography{main}

\end{document}